\documentclass[conference,letterpaper, 10pt]{IEEEtran}
\IEEEoverridecommandlockouts

\usepackage[utf8]{inputenc}
\usepackage[T1]{fontenc}
\usepackage[american]{babel}
\usepackage[center]{caption}
\usepackage{graphicx}
\usepackage{siunitx}
\usepackage{longtable,tabularx}
\usepackage{caption}
\usepackage{subfigure}
\usepackage{subcaption}
\usepackage{listings}
\usepackage{float}
\usepackage{amsmath,amssymb,exscale}
\usepackage{blindtext, graphicx}
\usepackage{verbatim}
\usepackage{algorithm}
\usepackage{algpseudocode}
\usepackage{fancyvrb}
\usepackage{bera}
\usepackage{amsmath}
\usepackage{amsthm}
\usepackage{amsfonts}
\newtheorem{problem}{Problem}
\newtheorem{theorem}{Theorem}
\usepackage{mathtools}
\usepackage{lipsum}
\usepackage[dvipsnames]{xcolor}
\usepackage{cuted}
\usepackage{flushend}
\usepackage{dblfloatfix} 
\usepackage{times}
\usepackage{placeins}

\newcommand{\flab}[1]{\label{fig:#1}}
\newcommand{\fig}[1]{Fig.\ref{fig:#1}}

\newcommand{\elab}[1]{\label{eqn:#1}}
\newcommand{\eqn}[1]{(\ref{eqn:#1})}

\usepackage{scalerel}

\usepackage[bookmarks=false, hidelinks]{hyperref}
\usepackage{cleveref}
\crefname{equation}{Eq.}{Eqs.} 

\crefname{figure}{Fig.}{Figs.}

\def\BibTeX{{\rm B\kern-.05em{\sc i\kern-.025em b}\kern-.08em
T\kern-.1667em\lower.7ex\hbox{E}\kern-.125emX}}

\makeatletter 
\let\old@ps@headings\ps@headings 
\let\old@ps@IEEEtitlepagestyle\ps@IEEEtitlepagestyle 
\def\confheader#1{%
\def\ps@headings{%
\old@ps@headings%
\def\@oddhead{\strut\hfill#1\hfill\strut}%
\def\@evenhead{\strut\hfill#1\hfill\strut}%
}%
\def\ps@IEEEtitlepagestyle{%
\old@ps@IEEEtitlepagestyle%
\def\@oddhead{\strut\hfill#1\hfill\strut}%
\def\@evenhead{\strut\hfill#1\hfill\strut}%
}%
\ps@headings%
} 
\makeatother 

\usepackage[letterpaper, margin=.75in]{geometry}
\newgeometry{top=1in, bottom=0.75in, left=0.75in, right=0.75in}

\begin{document}

\title{Sparse Actuation for LPV Systems with Full-State Feedback in $\mathcal{H}_2/\mathcal{H}_\infty$ Framework}

\author{\IEEEauthorblockN{Tanay Kumar\thanks{Graduate Student, \texttt{ktanay@tamu.edu}} \hspace{0.5in} Raktim Bhattacharya \thanks{Professor, \texttt{raktim@tamu.edu}}\vspace{.1in} }
\IEEEauthorblockA{Aerospace Engineering, Texas A\&M University,\\ College Station, TX, 77843-3141.
}}

\maketitle

\begin{abstract}
This paper addresses the sparse actuation problem for nonlinear systems represented in the Linear Parameter-Varying (LPV) form. We propose a convex optimization framework that concurrently determines actuator magnitude limits and the state-feedback law that guarantees a user-specified closed-loop performance in the $\mathcal{H}_2/\mathcal{H}_\infty$ sense. We also demonstrate that sparse actuation is achieved when the actuator magnitude-limits are minimized in the $l_1$ sense. This is the first paper that addresses this problem for LPV systems. The formulation is demonstrated in a vibration control problem for a flexible wing.
\end{abstract}


\section{Introduction}
Sparse actuation is beneficial in aerospace systems due to the need to reduce weight, complexity, and power consumption, which are critical factors in these applications. By utilizing a limited number of actuators, sparse actuation minimizes the system's overall mass and energy demands while maintaining control performance. Additionally, sparse actuation facilitates the development of more efficient control strategies, particularly in resource-constrained systems. Also, sparse actuation architectures are helpful in improving system resilience by identifying the minimum actuation capability needed to maintain operational integrity in the event of faults. This approach enables aerospace systems to carry out critical functions even if specific actuators fail, ensuring continuous operation and safety. By concentrating on minimal actuation requirements, engineers can develop fault-tolerant control strategies that activate alternative control paths or modes, maximizing the use of available resources. This method protects the system against unexpected failures and contributes to optimizing maintenance and repair processes, which are essential for long-duration aerospace missions.

Sparse actuation architectures are also crucial in controlling high-dimensional systems. For example, in fluid dynamics \cite{CHEN_ROWLEY_2011}, structural vibration \cite{hiramoto_optimal_2000}, and thermal processes \cite{deshpande_battery_LCSSCDC2021}, sparse actuation is essential due to the inherent complexity of these applications. In fluid dynamics, strategically placed actuators can effectively manipulate critical flow modes, eliminating the need for extensive instrumentation. For structural vibration control, a few strategically positioned actuators can specifically attenuate key vibration modes, reducing the overall actuator count. In thermal processes, employing sparse actuation enables focused control on essential areas, achieving efficient temperature management with fewer actuators while preserving desired system performance.

Determining the optimal placement of sensors and actuators in complex systems, along with the appropriate control law, is challenging. Typically, the actuation architecture (locations and magnitude limits) is determined ad hoc, which may limit the closed-loop performance. Therefore, developing an integrated approach that concurrently determines the placement of sensors and actuators and formulates the control law is crucial. This holistic method is key to realizing a resource-efficient closed-loop system that meets performance demands.

\subsection{Related Works}  
The issue of identifying a sparse set of actuators for closed-loop control in LTI systems has received significant attention. Traditional approaches to actuator selection in $\mathcal{H}_2/\mathcal{H}_\infty$   optimal state-feedback control often involve applying row-wise sparsity-inducing penalties to the controller gain matrix \cite{dhingra_admm_2014, zare_proximal_2020, munz_sensor_2014} or employing element-wise sparsity for structured controller designs \cite{jovanovic_controller_2016}. Additionally, integer programming and search algorithms have been utilized for addressing static output feedback control challenges \cite{nugroho_simultaneous_2018}, primarily aiming at system stabilization. Other methods consider the problem from the perspective of controllability and observability \cite{manohar_optimal_2020, summers_submodularity_2016}. The work in \cite{vedang} presents a convex optimization framework for designing full-state feedback controllers with bounds on actuator magnitude and rate limits.

\section{New Contributions}
This paper is the first to address the sparse actuation problem for \textit{nonlinear systems} represented in linear parameter varying (LPV) form. In contrast, previous research has focused exclusively on LTI systems. The formulation introduced here extends the framework in \cite{vedang} for LTI systems to LPV systems. This extension is particularly important for complex systems, such as aircraft, where dynamic properties vary with external factors like flight velocity and altitude.

We use the formulated theory to design a state-feedback controller to minimize the oscillations in a flexible wing due to external disturbances. Designing the actuator layout for such systems is nontrivial and often results in over-actuation or under-performing control systems. Using the proposed theory, we demonstrate how a state-feedback control can be designed with a nonlinear model to attenuate wing oscillations with a minimal number of actuators effectively. 

\section{Problem Formulation}
\subsection{Brief Overview of LPV Systems}
LPV control systems \cite{shamma1990analysis, shamma2012overview} represent an advanced control theoretical framework for controlling nonlinear systems through a set of linear controllers scheduled on real-time measurable parameters. The LPV approach evolves from traditional gain scheduling, which designs controllers at various fixed points and interpolates between them as parameters change, often without theoretical robustness guarantees. In contrast, the LPV framework incorporates parameter variations into the design, offering theoretical robustness guarantees.
Various methods for designing LPV controllers exist, including linear fractional transformations (LFT), single quadratic Lyapunov function (SQLF), and parameter-dependent quadratic Lyapunov function (PDQLF). These methods transform control design challenges into convex optimization tasks involving linear matrix inequalities (LMIs) \cite{el2000advances, helton2007linear}. A primary challenge within this framework is modeling a nonlinear system in LPV terms. For instance, the nonlinear equation \(\dot{x} = x^3\) can be represented as \(\dot{x} = A(\rho)x\) with \(A(\rho) = \rho\) and \(\rho(t) = x^2(t)\). When LPV models are affine in parameters, i.e., \(A(\rho) = A_0 + A_1\rho\), ensuring performance across all parameter variations is simpler, assuming \(\rho(t)\) lies within a convex polytope and constraints are applied at the polytope's vertices. However, if \(\rho(t)\)'s dependency is non-linear, solutions often involve randomized algorithms with probabilistic guarantees \cite{tempo2013randomized, fujisaki2003probabilistic}. Additional parameters are occasionally introduced to preserve an affine structure, increasing the computational demand as the problem size increases. The complexity of resolving linear matrix inequalities (LMIs) can rise to \(\mathcal{O}(n^6)\), where \(n\) represents the size of the problem. Additionally, constraining \(\rho(t)\) within a defined closed set presents its own set of challenges.

Despite these complexities, the LPV framework has seen significant success, particularly in the aerospace sector \cite{ganguli2002reconfigurable, balas2002linear, gilbert2010polynomial, marcos2009lpv}, demonstrating its utility in practical applications.

\subsection{Sparse Actuation for Full State-Feedback LPV Systems}
Let us consider the following LPV system,
\begin{equation}
\label{LPV}
\begin{aligned}
   &\dot{x}(t) = A(\rho)x(t) + B_u(\rho)u(t) + B_w(\rho)w(t), \\
   & z(t) = C_z(\rho)x(t) + D_u(\rho)u(t) + D_w(\rho)w(t),\\
   & u(t) = Kx(t),
\end{aligned}
\end{equation}
where $x\in\mathcal{R}^{N_x}$ is the state vector, $u\in\mathcal{R}^{N_u}$ is the control input vector, $w\in\mathcal{R}^{N_w}$ is the external disturbance that is always bounded, and $\rho$ is the time-varying parameter. Assuming that the matrix pair $(A(\rho), B_u(\rho))$ is stabilizable, we aim to design a full-state feedback controller $K$ such that the vector of regulated outputs $z\in\mathcal{R}^{N_z}$ is bounded under the influence of exogenous inputs $w$ and the control vector $u$ is sparse. Note that $u$ is the vector of all possible control inputs or actuators in the system. We can achieve a sparse actuation architecture by enforcing sparseness on $u$. 

The closed-loop system with state-feedback control law is given by
\begin{equation}
\label{close loop}
\begin{aligned}
   &\dot{x}(t) = (A(\rho) + B_u(\rho)K)x(t) + B_w(\rho)w(t), \\
   & z(t) = (C_z(\rho) + D_u(\rho)K)x(t) + D_w(\rho)w(t).
\end{aligned}
\end{equation}
The transfer function matrix from $w$ to $z$ of the closed-loop system in \cref{close loop}, which is represented by $\mathcal{G}_z(\rho,s)$, can be given as

\begin{align}
\mathcal{G}_z(\rho,s) &=  (C_z(\rho) + D_u(\rho)K)(sI-(A(\rho) \notag \\ & + B_u(\rho)K))^{-1}B_w(\rho) + D_w(\rho).
\end{align}
The performance criteria on $z$ can be defined in terms of the $\mathcal{H}_2$ or $\mathcal{H}_\infty$ norm of $\mathcal{G}_z(\rho,s)$. Specifically, $\|\mathcal{G}_z(\rho,s)\|_{\mathcal{H}_2}<\gamma_0$ or $\|\mathcal{G}_z(\rho,s)\|_{\mathcal{H}_\infty}<\gamma_0$ is required for some $\gamma_0>0$.

To achieve sparseness in the control input $u$, the $\mathcal{H}_2$ norm of the transfer function from $w$ to each $u_i$ is minimized. This transfer function, represented by $\mathcal{G}_{u_i}(\rho,s)$, corresponds to the $i^{th}$ component of $u=\begin{bmatrix}u_1 & \cdots &u_{N_u}\end{bmatrix}^\top$. Using \cref{LPV}, $\mathcal{G}_{u_i}(s)$ can be expressed as,
\begin{equation}
    \mathcal{G}_{u_i}(\rho,s) = \text{row}_i(K)(sI-(A(\rho) + B_u(\rho)K))^{-1}B_w(\rho),
\end{equation}
where $\text{row}_i(K)\in \mathcal{R}^{1\times N_x}$ is the $i^{th}$ row of the controller $K$. A smaller value of $\|\mathcal{G}_{u_i}(\rho,s)\|_{\mathcal{H}_2}$ indicates that the contribution of the $i^{th}$ actuator to the peak magitude of control is small. If this norm value falls below a certain threshold (e.g., $10^{-3}$), it implies that the $i^{th}$ actuator may be redundant and could be removed from the control architecture. In essence, a vector with sparse $\mathcal{H}_2$ norm values from $w$ to $u_i$ corresponds to an actuation structure with fewer active actuators. This is the basis for the sparse controller design approach discussed below.

If $\sqrt{\gamma_i}\geq0$ is an upper bound on $\|\mathcal{G}_{u_i}(\rho,s)\|_{\mathcal{H}_2}$, which can be written as
\begin{align*}
&    \begin{bmatrix}
            \|\mathcal{G}_{u_1}(\rho,s)\|_{\mathcal{H}_2} & \|\mathcal{G}_{u_2}(\rho,s)\|_{\mathcal{H}_2} & ... &\|\mathcal{G}_{u_{N_u}}(\rho,s)\|_{\mathcal{H}_2}
    \end{bmatrix}^\top, \\
&   \text{ or }     \begin{bmatrix}
            \sqrt{\gamma_1} & \sqrt{\gamma_2} & ... & \sqrt{\gamma_{N_u}}
        \end{bmatrix}^\top =: \sqrt{\Gamma}    
\end{align*}

To achieve sparseness in $\Gamma$, we minimize the weighted $l_1$ norm of $\Gamma$ which is
\begin{equation}
    \|\Gamma\|_{l_1,\alpha} := \alpha^\top|\Gamma|,
\end{equation}
where $\alpha$ is the weight vector. Based on this, the control synthesis problem can be mathematically stated as,

\begin{problem}
    The $\mathcal{H}_\infty$ full-state feedback control problem can be expressed as the following optimization problem.
    \begin{equation}
    \left.
    \begin{aligned}
    & \min _{\Gamma, K}\|\Gamma\|_{l_1,\alpha}, \\
    \text { subject to } &              \left\|\mathcal{G}_z(\rho,s)\right\|_{\mathcal{H}_{\infty}}      <\gamma_0, \\
    & \left\|\mathcal{G}_{u_i}(\rho,s)\right\|_{\mathcal{H}_2}      <\sqrt{\gamma_i}, \quad i=1,2, \cdots, N_u.
\end{aligned}\right\}
\elab{prob hinf}
\end{equation}
\end{problem}

\begin{theorem}
\label{theorem 1}
   The solution to $\mathcal{H}_{\infty}$ full state-feedback design problem defined in \eqn{prob hinf} is obtained by solving the following optimization problem, where the controller gain is given by $K=W X^{-1}$.

\begin{equation}\left.
    \begin{aligned}   
    & \min _{\Gamma>0, X>0, W}\|\Gamma\|_{l_1,\alpha} \text { subject to } \\
     &\begin{bmatrix}
    M_{11}(\rho) & B_w(\rho) & \left(C_z(\rho) X(\rho)+D_u(\rho) W(\rho)\right)^{\top} \\
    * & -\gamma_0 I & D_w(\rho)^{\top} \\
    *  & * & -\gamma_0 I
     \end{bmatrix}<0,\\
    & M_{11}(\rho)+B_w(\rho) B_w(\rho)^{\top}<0, \\
    & \begin{bmatrix}
    -\gamma_i & \operatorname{row}_i(W(\rho)) \\
    * & -X(\rho)
    \end{bmatrix}<0, \quad i=1,2, \cdots, N_u;
    \end{aligned}\right\}
    \elab{hinf opt}
\end{equation} 
\end{theorem}
where
\begin{align*}
M_{11}(\rho) := \operatorname{sym}\left(A(\rho) X(\rho)+B_u(\rho) W(\rho)\right). 
\end{align*}
and $\operatorname{sym}(.):= \frac{1}{2}[(.)+(.)^\top]$
\begin{proof}
Follows the steps for theorem 1 in \cite{vedang}.
\end{proof}

Alternatively, the performance of the closed-loop controller can be characterized using the $\mathcal{H}_2$ norm, as previously mentioned. The corresponding synthesis problem can be formulated as follows.
\begin{problem}
    The $\mathcal{H}_2$ full-state feedback control problem can be expressed as the following optimization problem.
    \begin{equation}
    \begin{aligned}
    & \min _{\Gamma, K}\|\Gamma\|_{l_1,\alpha} \\
    \text { subject to } &              \left\|\mathcal{G}_z(\rho,s)\right\|_{\mathcal{H}_2}      <\gamma_0 \\
    & \left\|\mathcal{G}_{u_i}(\rho,s)\right\|_{\mathcal{H}_2}      <\sqrt{\gamma_i}, \quad i=1,2, \cdots, N_u.   
\end{aligned}
    \elab{prob h2}
\end{equation}
\end{problem}

\begin{theorem}
\label{theorem 2}
   The solution to $\mathcal{H}_2$ full state-feedback design problem defined in \eqn{prob h2} is obtained by solving the following optimization problem, where the controller gain is given by $K=W X^{-1}$.
\begin{equation}\left.
    \begin{aligned}   
    & \min _{\Gamma>0, X>0, W}\|\Gamma\|_{l_1,\alpha} \text { subject to } \\
    & \begin{bmatrix}
    -Z(\rho) & C_z(\rho)X(\rho) +D_u(\rho)W(\rho) \\
    * & -X(\rho)
     \end{bmatrix}<0, \\
    & \operatorname{sym}\left(A(\rho) X(\rho)+B_u(\rho) W(\rho)\right)+B_w(\rho) B_w(\rho)^{\top}<0 \text {, } \\
    & \operatorname{tr}(Z(\rho))<\gamma_0^2\\
    & \begin{bmatrix}
    -\gamma_i & W(\rho) \\
    * & -X(\rho)
    \end{bmatrix}<0, \quad i=1,2, \cdots, N_u
    \end{aligned}\right\}
    \elab{h2 opt}
    \end{equation} 
\end{theorem}

\begin{proof}
Follows the steps for theorem 2 in \cite{vedang}.
\end{proof}

\subsection{Iterative Reweighted $l_1$ Minimization}
We apply the iterative reweighting approach described in \cite{candes2008enhancing} to obtain a sparse actuation configuration. This involves repeatedly solving the optimization problems \eqn{hinf opt} or \eqn{h2 opt}, with the weights for the $(j+1)^{th}$ iteration calculated based on the previous iteration as:
\begin{equation}
    \alpha_i^{j+1} = \left( \epsilon + |\Gamma_i^j|\right)^{-1},
\end{equation}
where $\epsilon>0$ is a small value to ensure that the weights remain well-defined. The initial weight vector is considered equal to 1 to maintain generality. The iteration process terminates when either the convergence criterion is satisfied, or the maximum number of iterations is reached. The final solution is refined by eliminating the actuators with very small magnitudes and re-solving the optimization problem in \eqn{hinf opt} or \eqn{h2 opt} with equal weights.

\subsection{System Model}
In this section, we apply the principles outlined in Theorems 1 and 2 to reduce oscillations in an aircraft wing that experiences external disturbances. The wing is modeled as a series of bars connected longitudinally by torsional springs, as shown in \fig{diagram}. Each bar is equipped with an actuator capable of exerting force, such as through aerodynamic control surfaces. We model disturbances as forces impacting each bar. Both control and disturbance forces are assumed to be point forces acting perpendicularly at the center of mass of each bar. We have omitted the effects of gravity from the dynamics to simplify the LPV problem formulation. Forces due to gravity can be considered as external disturbances. The restoring torque of the torsional spring is defined as \(\tau = -k_1\theta - k_2\theta^3\), where \(\theta\) represents the angular deviation from the spring's rest position, assumed to be \(0^\circ\) in this analysis.
\begin{figure}[h!]
\centering
  \includegraphics[scale= 1]{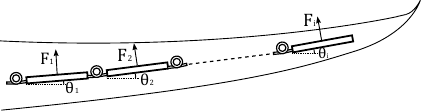}
  \caption{Schematic of a flexible wing, modeled as a series of bars and rotational springs.}
  \flab{diagram}
\end{figure}
The degrees of freedom of the system are defined by the angles each cantilever makes with the horizontal ($\theta_i$). Therefore, the system has $N_u$ degrees of freedom. Correspondingly the state vector $x\in\mathcal{R}^{2N_u}$ consists of angles ($\theta$) and angular velocities ($\dot{\theta}$). Each cantilever can produce a force normal to its plane so the control input is $u\in\mathcal{R}^{N_u}$. The wind disturbances on the wing are in the form of forces normal to the wing surface. Therefore, disturbance acts along the same direction as the control input $u$. The dynamics of the system under consideration are derived using the Lagrangian mechanics. The objective is to develop a feedback control law with a minimum number of actuated cantilevers so that the wing experiences minimal oscillations under unsteady wind loading. It is important to emphasize that while it is feasible to design a controller that maintains the wing's angle without any alterations, such an approach is strongly discouraged from a structural standpoint due to the potential for excessive strain. Consequently, our objective is focused solely on reducing oscillations in response to external disturbances.

The equation of motion (EOM) derived is non-linear in $\theta$. Therefore, the system can be converted to LPV form, with time-varying parameters $\rho$ that are the function of the states of the system itself, i.e. $\rho_i = f(\theta)$. Such LPV models are called \textit{quasi}-LPV models.

For a system with $n$ masses, the LPV system parameters ($\rho$) are chosen as:

\begin{equation}
    \begin{aligned}
        & \rho_i = \theta_i^2, \quad i=1,2, \cdots, N_x\\
    \end{aligned}
\end{equation}

The EOM for this problem is nonlinear as the states appear in quadratic form, which allows for quasi-LPV representation. The resulting LPV system has $N_x$ parameters and is affine in nature. The physical parameters of the system are given in \autoref{tab:Params}.

\begin{table}[ht]
\caption{Physical Parameters.}
\centering
\begin{tabular}{|c|c|}
  \hline
  \textbf{Parameter} & \textbf{Value} \\ \hline
    $m$ (cantilever mass) & 1.5 kg  \\ \hline
    $n$ (number of cantilevers) & 5  \\ \hline
    $l$ (length of a cantilever) & 1 m \\ \hline
    $k_1$ & 10 N-m/rad\\ \hline
    $k_2$ & 1.5 N-m/rad$^3$ \\ \hline
  \end{tabular}
\label{tab:Params}
\end{table}

\section{Results}
\label{Results}

In this section, we discuss the controller performance from the $\mathcal{H}_2$ and $\mathcal{H}_\infty$ optimal designs. The feedback controller is derived based on two slightly different system models. The first model assumes a linear spring (where $k_2=0$), leading to an LTI system. The second model accounts for a nonlinear spring, resulting in an LPV system. The controller's performance is validated by testing it on nonlinear system dynamics subject to external disturbances. In the following figures, the nonlinear simulation results of the closed-loop system are presented, where the controller designed based on the LTI system model is labeled as the `LTI Model', and the controller designed using the LPV system model is labeled as the `LPV Model'. The disturbance is introduced through a non-zero initial condition and a slowly varying wind gust, represented as $d(t) = 0.3 wn(t) + \sin{0.005t}$, where $wn$ is white noise with a uniform distribution.

\subsection{Sparse Actuation with $\mathcal{H}_\infty$ Full-State Feedback}
Here, we address the $\mathcal{H}_\infty$ full-state feedback control design aimed at constraining the perturbations in the wing shown in \cref{fig:diagram} caused by wind disturbances. \cref{fig:Hinf_theta_UB14} and \cref{fig:Hinf_theta_dot_UB14} present a comparison between the state trajectories of the open-loop and closed-loop simulation. A stable closed-loop performance can be noticed as opposed to the open-loop case. Additionally, significant difference can be seen between the closed-loop response of LTI and LPV models as the LTI models has higher overshoot and settling time. 

\begin{figure}[ht]
\centering
  \includegraphics[scale= 0.425]{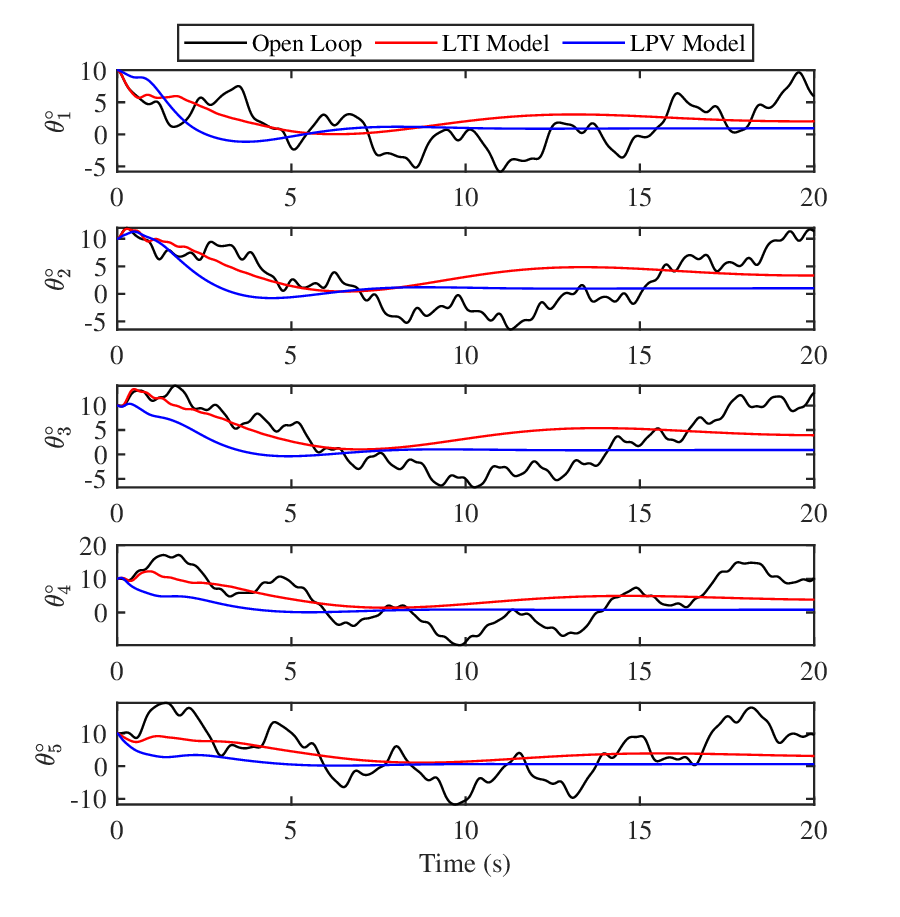}
  \caption{Angular deflection trajectories from nonlinear simulations for the given $\mathcal{H}_\infty$ performance $\gamma_0=0.15$.}
  \label{fig:Hinf_theta_UB14}
\end{figure}

\begin{figure}[ht]
\centering
  \includegraphics[scale= 0.425]{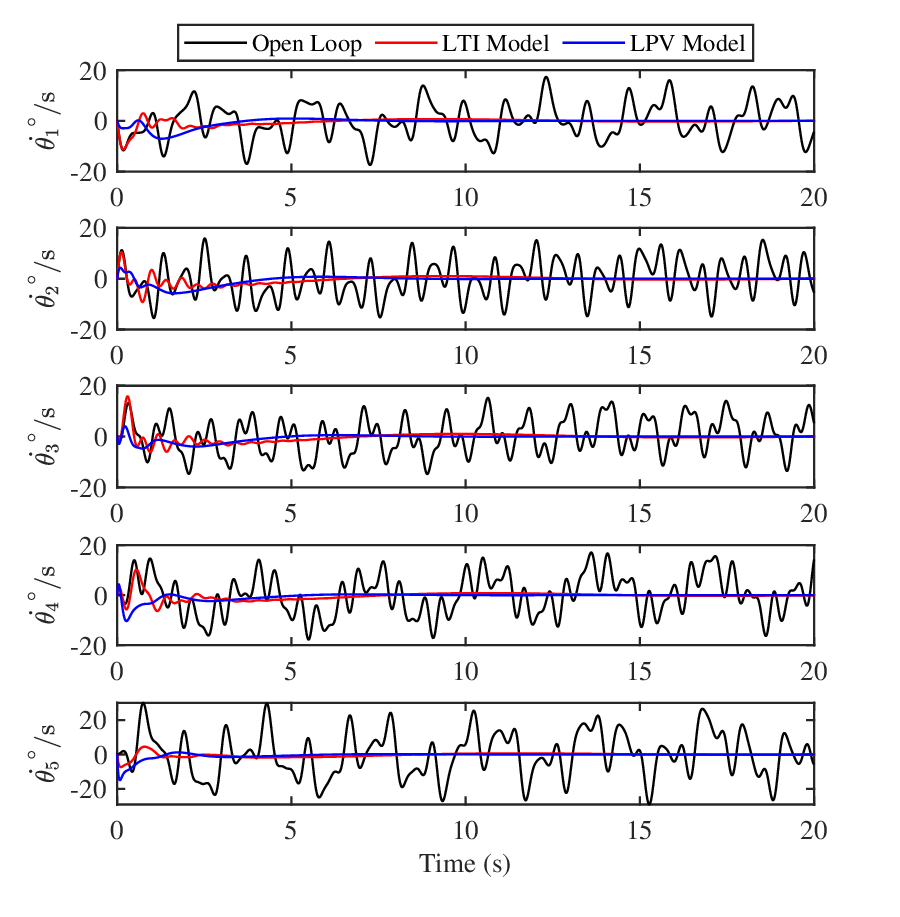}
  \caption{Angular deflection rate trajectories from nonlinear simulations for the given $\mathcal{H}_\infty$ performance $\gamma_0=0.15$.}
  \label{fig:Hinf_theta_dot_UB14}
\end{figure}

\cref{fig:Hinf_gamma0_UB14} and \cref{fig:Hinf_gamma0_UB8} illustrate the $\|u_i(t)\|_\infty$ achieved from the $\mathcal{H}_\infty$ optimal design for a specific value of $\gamma_0$ and varying values of $\gamma_{UB}$. The parameter $\gamma_{UB}$ imposes an upper bound on the control input to emulate hardware limitations. We observe from \cref{fig:Hinf_gamma0_UB14}, that for the controller derived using LPV system mode, actuators 1 is not required to achieve the desired closed-loop disturbance attenuation due to a negligible value of $\|u_i(t)\|_\infty$. Additionally, the control inputs of actuators 2, 3, and 4 are relatively small. Although the LTI model results in a higher degree of sparseness, the controller performance is subpar as discussed in \cref{fig:Hinf_theta_UB14} and \cref{fig:Hinf_theta_dot_UB14} as the controller does not account for the  underlying nonlinearities in the system dynamics..

\begin{figure}[ht]
\centering
  \includegraphics[scale= 0.5]{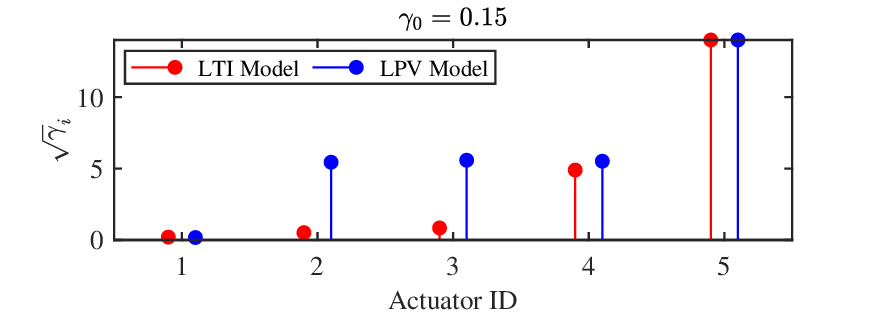}
  \caption{Minimum control efforts from nonlinear simulations for a given $\mathcal{H}_\infty$ performance with $\sqrt{\gamma_{UB}} = 14$.}
  \label{fig:Hinf_gamma0_UB14}
\end{figure}

As more stringent constraints are placed on the upper bound $\gamma_{UB}$, the control distribution adapts accordingly as expected. As observed in \cref{fig:Hinf_gamma0_UB8}, each actuator has an increased contribution of control inputs to meet similar performance requirements as that in the previous case. This highlights the inherent trade-off between actuator sparsity and the maximum allowable control magnitude.

\begin{figure}[ht]
\centering
  \includegraphics[scale= 0.5]{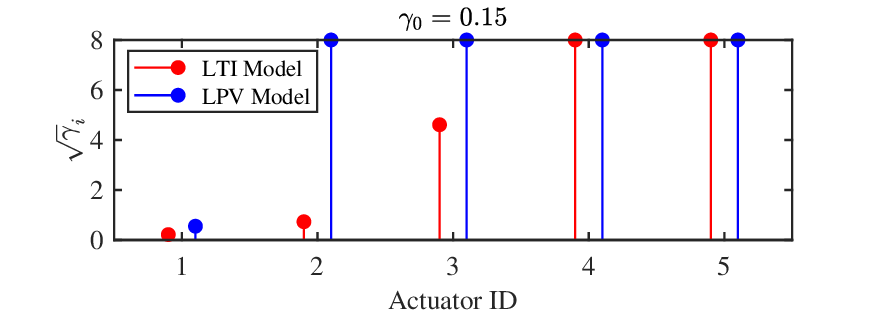}
  \caption{Minimum control efforts from nonlinear simulations for a given $\mathcal{H}_\infty$ performance with $\sqrt{\gamma_{UB}} = 8$.}
  \label{fig:Hinf_gamma0_UB8}
\end{figure}

\subsection{Sparse Actuation with $\mathcal{H}_2$ Full-State Feedback}

Here, we address the $\mathcal{H}_2$ full-state feedback disturbance rejection problem. \cref{fig:H2_theta_UB8} and \cref{fig:H2_theta_dot_UB8} present a comparison between the state trajectories of the open-loop and closed-loop simulations. It is observed that the settling time and overshoot for the system is higher than that in the case of $\mathcal{H}_\infty$ optimal control. This is because, unlike $\mathcal{H}_2$ formulation, $\mathcal{H}_\infty$ explicitly considers the worst-case disturbances and provides more robust performance, which inherently results in less oscillatory behavior under external perturbations or uncertainties.

\begin{figure}[ht]
\centering
  \includegraphics[scale= 0.425]{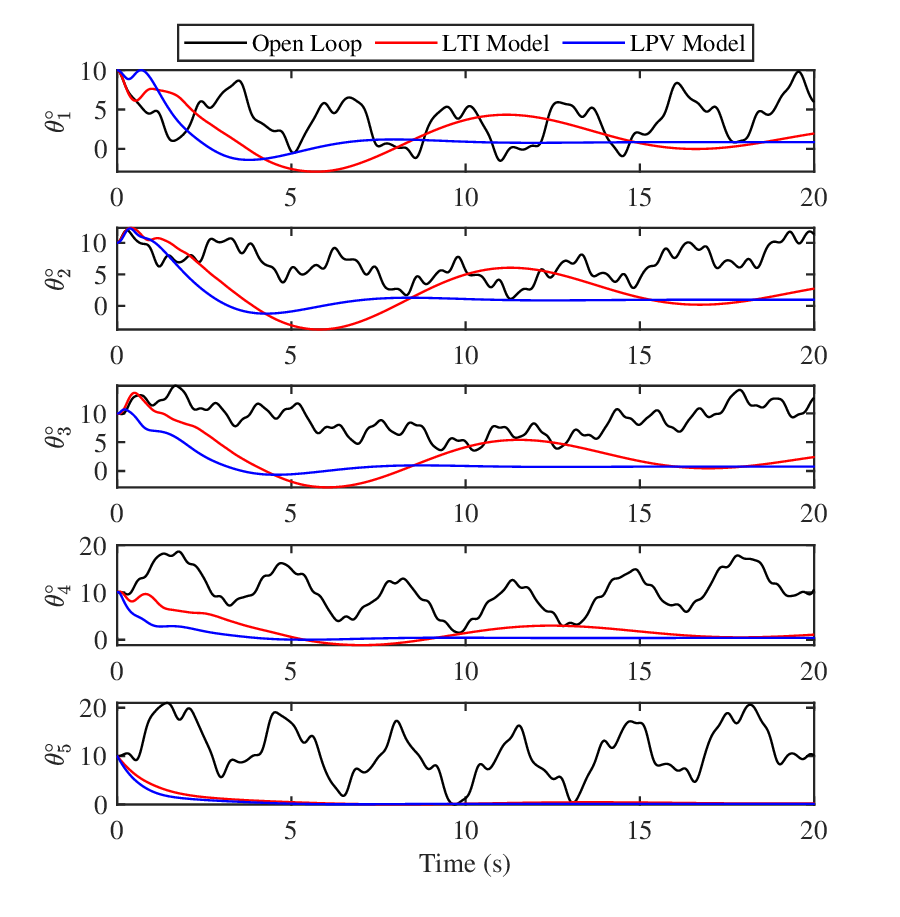}
  \caption{Angular deflection trajectories from nonlinear simulations for the given $\mathcal{H}_2$ performance $\gamma_0=0.15$.}
  \label{fig:H2_theta_UB8}
\end{figure}

\begin{figure}[ht]
\centering
  \includegraphics[scale= 0.45]{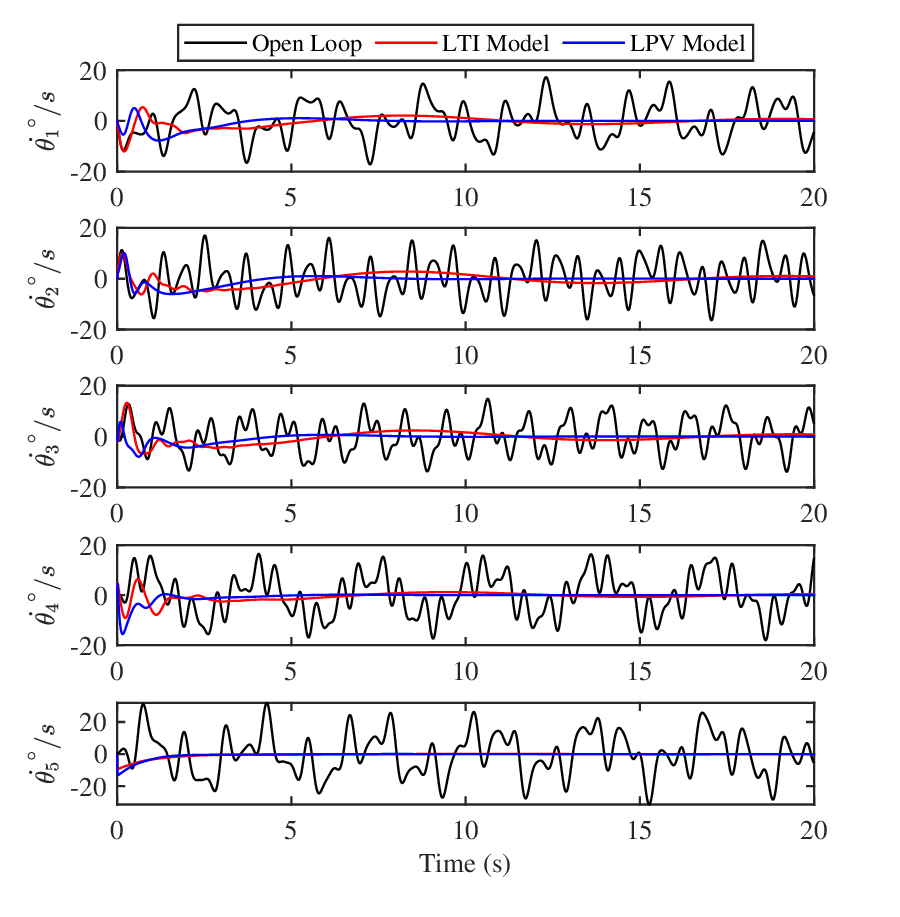}
  \caption{Angular deflection rate trajectories from nonlinear simulations for the given $\mathcal{H}_2$ performance $\gamma_0=0.15$.}
  \label{fig:H2_theta_dot_UB8}
\end{figure}

\cref{fig:H2_gamma0_UB14} and \cref{fig:H2_gamma0_UB8} illustrate the $\|u_i(t)\|_\infty$ achieved from the $\mathcal{H}_2$ optimal design for a specific value of $\gamma_0$ and varying values of $\gamma_{UB}$. We observe from \cref{fig:H2_gamma0_UB14}, for $\mathcal{H}_2$ case, that actuators 4, and 5 contribute majorly to the control inputs to achieve the desired closed-loop disturbance attenuation. The $\|u_i(t)\|_\infty$ for the remaining actuators is significantly lower.

\begin{figure}[ht]
\centering
  \includegraphics[scale= 0.5]{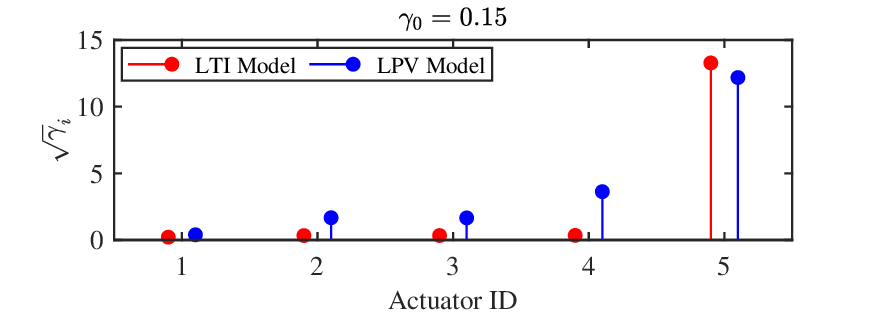}
  \caption{Minimum control efforts from nonlinear simulations for a given $\mathcal{H}_2$ performance with $\sqrt{\gamma_{UB}} = 14$.}
  \label{fig:H2_gamma0_UB14}
\end{figure}

With more stringent constraints imposed on the upper bound $\gamma_{UB}$, the control distribution adjusts accordingly, similar to $\mathcal{H}_\infty$ case. As observed in \cref{fig:H2_gamma0_UB8}, no additional actuators are required to achieve the same level of performance. However, the contribution of only actuator 4 increases significantly. This is due to the fact that $\mathcal{H}_2$ control focuses on minimizing the mean-square response of the system output to disturbances, aiming to reduce the average energy of the output signal. Consequently, while a higher degree of actuation sparsity is observed, it comes at the cost of a longer settling time.

\begin{figure}[ht]
\centering
  \includegraphics[scale= 0.5]{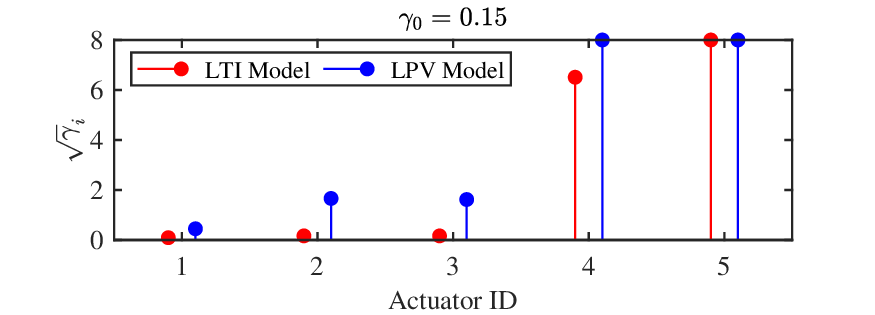}
  \caption{Minimum control efforts from nonlinear simulations for a given $\mathcal{H}_2$ performance with $\sqrt{\gamma_{UB}} = 8$.}
  \label{fig:H2_gamma0_UB8}
\end{figure}


\section{Conclusions}
In conclusion, this paper introduced novel convex optimization formulations for the codesign of actuation architecture and control laws, guaranteeing the specified $\mathcal{H}_2/\mathcal{H}_\infty$ closed-loop performance in nonlinear systems modeled in LPV form. Detailed in \cref{theorem 1} and \cref{theorem 2}, these formulations were applied to the vibration control of a flexible wing modeled nonlinearly. Our approach could effectively identify a sparse actuation architecture by minimizing the one-norm of actuator magnitude limits while maintaining desired closed-loop performance. This was verified through nonlinear simulations, confirming that actuators with zero magnitude limits could be eliminated without performance loss. Furthermore, our comparisons revealed that sparse solutions based on LTI models were inferior to those employing LPV models, highlighting the importance of incorporating nonlinearities within the design framework.

\bibliographystyle{IEEEtran}
\bibliography{refs,refs2}

\begin{thebibliography}{10}
\providecommand{\url}[1]{#1}
\csname url@samestyle\endcsname
\providecommand{\newblock}{\relax}
\providecommand{\bibinfo}[2]{#2}
\providecommand{\BIBentrySTDinterwordspacing}{\spaceskip=0pt\relax}
\providecommand{\BIBentryALTinterwordstretchfactor}{4}
\providecommand{\BIBentryALTinterwordspacing}{\spaceskip=\fontdimen2\font plus
\BIBentryALTinterwordstretchfactor\fontdimen3\font minus \fontdimen4\font\relax}
\providecommand{\BIBforeignlanguage}[2]{{%
\expandafter\ifx\csname l@#1\endcsname\relax
\typeout{** WARNING: IEEEtran.bst: No hyphenation pattern has been}%
\typeout{** loaded for the language `#1'. Using the pattern for}%
\typeout{** the default language instead.}%
\else
\language=\csname l@#1\endcsname
\fi
#2}}
\providecommand{\BIBdecl}{\relax}
\BIBdecl

\bibitem{CHEN_ROWLEY_2011}
K.~K. CHEN and C.~W. ROWLEY, ``$\mathcal{H}_2$ optimal actuator and sensor placement in the linearised complex ginzburg–landau system,'' \emph{Journal of Fluid Mechanics}, vol. 681, p. 241–260, 2011.

\bibitem{hiramoto_optimal_2000}
\BIBentryALTinterwordspacing
K.~Hiramoto, H.~Doki, and G.~Obinata, ``\BIBforeignlanguage{en}{{Optimal Sensor/Actuator Placement for Active Vibration Control Using Explicit Solution of Algebraic {Riccati} Equation}},'' \emph{\BIBforeignlanguage{en}{Journal of Sound and Vibration}}, vol. 229, no.~5, pp. 1057--1075, Feb. 2000. [Online]. Available: \url{http://www.sciencedirect.com/science/article/pii/S0022460X99925300}
\BIBentrySTDinterwordspacing

\bibitem{deshpande_battery_LCSSCDC2021}
V.~M. Deshpande, R.~Bhattacharya, and K.~Subbarao, ``{Sensor Placement With Optimal Precision for Temperature Estimation of Battery Systems},'' \emph{IEEE Control Systems Letters}, vol.~6, pp. 1082--1087, 2022.

\bibitem{dhingra_admm_2014}
N.~K. Dhingra, M.~R. Jovanović, and Z.~Luo, ``{An {ADMM} Algorithm for Optimal Sensor and Actuator Selection},'' in \emph{53rd {IEEE} {Conference} on {Decision} and {Control}}, Dec. 2014, pp. 4039--4044, iSSN: 0191-2216.

\bibitem{zare_proximal_2020}
A.~Zare, H.~Mohammadi, N.~K. Dhingra, T.~T. Georgiou, and M.~R. Jovanović, ``Proximal {Algorithms} for {Large}-{Scale} {Statistical} {Modeling} and {Sensor}/{Actuator} {Selection},'' \emph{IEEE Transactions on Automatic Control}, vol.~65, no.~8, pp. 3441--3456, Aug. 2020.

\bibitem{munz_sensor_2014}
U.~Münz, M.~Pfister, and P.~Wolfrum, ``Sensor and {Actuator} {Placement} for {Linear} {Systems} {Based} on $\mathcal{H}_2$ and $\mathcal{H}_{\infty}$ {Optimization},'' \emph{IEEE Transactions on Automatic Control}, vol.~59, no.~11, pp. 2984--2989, Nov. 2014.

\bibitem{jovanovic_controller_2016}
\BIBentryALTinterwordspacing
M.~R. Jovanović and N.~K. Dhingra, ``\BIBforeignlanguage{en}{{Controller Architectures: {Tradeoffs} Between Performance and Structure}},'' \emph{\BIBforeignlanguage{en}{European Journal of Control}}, vol.~30, pp. 76--91, 2016. [Online]. Available: \url{http://www.sciencedirect.com/science/article/pii/S0947358016300334}
\BIBentrySTDinterwordspacing

\bibitem{nugroho_simultaneous_2018}
S.~Nugroho, A.~F. Taha, T.~Summers, and N.~Gatsis, ``Simultaneous {Sensor} and {Actuator} {Selection}/{Placement} through {Output} {Feedback} {Control},'' in \emph{2018 {Annual} {American} {Control} {Conference} ({ACC})}, 2018, pp. 4159--4164.

\bibitem{manohar_optimal_2020}
K.~Manohar, J.~N. Kutz, and S.~L. Brunton, ``{Optimal Sensor and Actuator Selection Using Balanced Model Reduction},'' \emph{IEEE Transactions on Automatic Control}, vol.~67, no.~4, pp. 2108--2115, 2022.

\bibitem{summers_submodularity_2016}
T.~H. Summers, F.~L. Cortesi, and J.~Lygeros, ``On {Submodularity} and {Controllability} in {Complex} {Dynamical} {Networks},'' \emph{IEEE Transactions on Control of Network Systems}, vol.~3, no.~1, pp. 91--101, 2016.

\bibitem{vedang}
\BIBentryALTinterwordspacing
V.~M. Deshpande and R.~Bhattacharya, ``$\mathcal{H}_2/\mathcal{H}_\infty$ optimal control with sparse sensing and actuation,'' 2024. [Online]. Available: \url{https://arxiv.org/abs/2409.09596}
\BIBentrySTDinterwordspacing

\bibitem{shamma1990analysis}
J.~S. Shamma and M.~Athans, ``Analysis of gain scheduled control for nonlinear plants,'' \emph{IEEE Transactions on Automatic Control}, vol.~35, no.~8, pp. 898--907, 1990.

\bibitem{shamma2012overview}
J.~S. Shamma, ``An overview of lpv systems,'' \emph{Control of linear parameter varying systems with applications}, pp. 3--26, 2012.

\bibitem{el2000advances}
L.~El~Ghaoui and S.-l. Niculescu, \emph{{Advances in Linear Matrix Inequality Methods in Control}}.\hskip 1em plus 0.5em minus 0.4em\relax SIAM, 2000.

\bibitem{helton2007linear}
J.~W. Helton and V.~Vinnikov, ``Linear matrix inequality representation of sets,'' \emph{Communications on Pure and Applied Mathematics: A Journal Issued by the Courant Institute of Mathematical Sciences}, vol.~60, no.~5, pp. 654--674, 2007.

\bibitem{tempo2013randomized}
R.~Tempo, G.~Calafiore, F.~Dabbene \emph{et~al.}, \emph{Randomized algorithms for analysis and control of uncertain systems: with applications}.\hskip 1em plus 0.5em minus 0.4em\relax Springer, 2013, vol.~7.

\bibitem{fujisaki2003probabilistic}
Y.~Fujisaki, F.~Dabbene, and R.~Tempo, ``Probabilistic design of lpv control systems,'' \emph{Automatica}, vol.~39, no.~8, pp. 1323--1337, 2003.

\bibitem{ganguli2002reconfigurable}
S.~Ganguli, A.~Marcos, and G.~Balas, ``Reconfigurable lpv control design for boeing 747-100/200 longitudinal axis,'' in \emph{Proceedings of the 2002 American control conference (IEEE cat. no. CH37301)}, vol.~5.\hskip 1em plus 0.5em minus 0.4em\relax IEEE, 2002, pp. 3612--3617.

\bibitem{balas2002linear}
G.~J. Balas, ``Linear, parameter-varying control and its application to aerospace systems,'' in \emph{ICAS congress proceedings}, 2002, pp. 541--1.

\bibitem{gilbert2010polynomial}
W.~Gilbert, D.~Henrion, J.~Bernussou, and D.~Boyer, ``Polynomial lpv synthesis applied to turbofan engines,'' \emph{Control engineering practice}, vol.~18, no.~9, pp. 1077--1083, 2010.

\bibitem{marcos2009lpv}
A.~Marcos and S.~Bennani, ``Lpv modeling, analysis and design in space systems: Rationale, objectives and limitations,'' in \emph{AIAA guidance, navigation, and control conference}, 2009, p. 5633.

\bibitem{candes2008enhancing}
E.~J. Candes, M.~B. Wakin, and S.~P. Boyd, ``Enhancing sparsity by reweighted $l_1$ minimization,'' \emph{Journal of Fourier analysis and applications}, vol.~14, pp. 877--905, 2008.

\end{thebibliography}

\end{document}